\documentclass[12pt]{article}
\usepackage{amsmath}
\usepackage{amssymb}
\usepackage{amsthm}
\usepackage{fullpage}
\usepackage[utf8]{inputenc}
\usepackage{mathtools}
\usepackage{url}
\usepackage[usenames,svgnames]{xcolor}
\usepackage{cite}

\usepackage{fullpage}
\usepackage[utf8]{inputenc}
\usepackage{mathtools}
\usepackage{url}
\usepackage[usenames,svgnames]{xcolor}
\usepackage{cite}
\usepackage{enumi tem}
\usepackage{datetime}
\usepackage[titletoc,title]{appendix}
\usepackage{comment}
\usepackage{mathdots}
\usepackage{graphicx}

\usepackage[pagebackref=false]{hyperref} 
\usepackage[all]{hypcap} 
\RequirePackage{color}

\usepackage{amsmath,amssymb,amsthm,tikz}
\usetikzlibrary{decorations.pathreplacing}
\usepackage{bm}
\usepackage{caption}
\usepackage[final,color]{showkeys}
\usepackage{xr}

\usepackage{mdwlist}	

\usepackage{placeins}

\usepackage{nicefrac}	

\theoremstyle{plain}
\newtheorem{theorem}{\sc Theorem}[section]
\newtheorem{lemma}[theorem]{\sc Lemma}

\newtheorem{proposition}[theorem]{\sc Proposition}
\newtheorem{corollary}[theorem]{\sc Corollary}
\newtheorem{example}{Example}[section]
\newtheorem{examples}{Example}[subsection]

\newtheorem{remark}{Remark}[section]

\theoremstyle{definition}
\newtheorem{definition}{Definition}[section]

\numberwithin{equation}{section} 

\numberwithin{equation}{section}
\theoremstyle{remark}

\newcommand{\pmeas}[2]{\xi_{#2}^{(#1)}}
\newcommand{\pmppre}{\wt{\xi}^{(n,d)}_{E'(q)}}
\newcommand{\pmp}{\pmeas{d}{E'(q)}}
\newcommand{\pme}{\pmeas{d}{E(q)}}
\newcommand{\pmh}{\pmeas{d}{H(q)}}
\newcommand{\Pmeas}[3]{\theta_{#3}^{(#2,#1)}}
\newcommand{\Pmp}{\Pmeas{d}{n}{E'(q)}}
\newcommand{\Pme}{\Pmeas{d}{n}{E(q)}}
\newcommand{\Pmh}{\Pmeas{d}{n}{H(q)}}

\newcommand{\Pf}[3]{\wt{Z}_{#3}^{(#2,#1)}}
\newcommand{\Pfp}{\Pf{d}{n}{E'(q)}}

\newcommand{\pf}[2]{Z_{#2}^{(#1)}}
\newcommand{\pfp}{\pf{d}{E'(q)}}

\newcommand{\ip}[1]{\left\langle #1 \right\rangle}

\DeclareMathOperator{\Li}{Li}
\DeclareMathOperator{\aut}{aut}

\def\ra{{\rightarrow}}

\definecolor{my-blue}{rgb}{0.0,0.0,0.6}
\definecolor{my-red}{rgb}{0.5,0.0,0.0}
\definecolor{my-green}{rgb}{0.0,0.5,0.0}
\definecolor{nicos-red}{rgb}{0.75,0.0,0.0}
\definecolor{light-gray}{gray}{0.6}
\definecolor{really-light-gray}{gray}{0.8}

\newcommand{\x}{\mathcal {X}}

\def\be{\begin{equation}}
\def\ee{\end{equation}}

\def\bea{\begin{eqnarray}}
\def\eea{\end{eqnarray}}

\def\bt{\begin{theorem}}
\def\et{\end{theorem}}

\def\bex{\begin{example}\small \rm}
\def\eex{\end{example}}

\def\bexs{\begin{examples}\small \rm}
\def\eexs{\end{examples}}

\def\ra{\rightarrow}

\def\deq{\coloneqq}

\def\br{\begin{remark}\small \rm}
\def\er{\end{remark}}

\def\&{&{\hskip -20pt}}

\def\Ib{\mathbf{I}}

\def\Pb{\mathbf{P}}

\def\Zb{\mathbf{Z}}

\def\Nbb{\mathbb{N}}

\def\Rbb{\mathbb{R}}

\def\Zbb{\mathbb{Z}}

\def\Nbb{\mathbb{N}}
\def\Rbb{\mathbb{R}}
\def\Zbb{\mathbb{Z}}

\def\ep{\epsilon}

\let\Oldsection\section
\renewcommand{\section}{\FloatBarrier\Oldsection}

\let\Oldsubsection\subsection
\renewcommand{\subsection}{\FloatBarrier\Oldsubsection}

\let\Oldsubsubsection\subsubsection
\renewcommand{\subsubsection}{\FloatBarrier\Oldsubsubsection}

\newcommand{\wt}[1]{\widetilde{#1}}

\def\bN{\mathbb{N}}
\def\mP{\mathcal{P}}

\newcommand{\rb}[1]{\left(#1\right)}
\newcommand{\ab}[1]{\left[#1\right]}
\newcommand{\abs}[1]{\left|#1\right|}

\newcommand{\set}[1]{\left\{#1\right\}}

\newcommand{\lambdamap}{\Lambda^{(n)}_d}

\newcommand{\y}{\mathcal{Y}}

\newcommand{\rr}{\mathbb{R}}

\definecolor{darkgreen}{rgb}{0.0,0.5,0.0}
\definecolor{darkblue}{rgb}{0.0,0.0,0.3}
\definecolor{nicosred}{rgb}{0.65,0.1,0.1}
\definecolor{light-gray}{gray}{0.7}
\allowdisplaybreaks[1]

\newcommand{\fM}{\mathfrak{M}}

\begin{document}
\baselineskip 16pt
\medskip
\begin{center}
\begin{Large}\fontfamily{cmss}
\fontsize{17pt}{27pt}
\selectfont
\textbf{Asymptotics of quantum weighted Hurwitz numbers}
\end{Large}\\
\bigskip
\begin{large}  {J. Harnad}$^{1,2}$ and {Janosch Ortmann}$^{1,2}$ 
 \end{large}
\\
\bigskip
\begin{small}
$^{1}${\em Centre de recherches math\'ematiques,
Universit\'e de Montr\'eal\\ C.~P.~6128, succ. centre ville, Montr\'eal,
QC, Canada H3C 3J7 } \\
\smallskip
$^{2}${\em Department of Mathematics and
Statistics, Concordia University\\ 1455 de Maisonneuve Blvd.~W.  
Montr\'eal, QC,  Canada H3G 1M8 } 
\end{small}
\end{center}
\bigskip

\begin{abstract}

   This work concerns both the semiclassical and zero temperature asymptotics of quantum weighted double
   Hurwitz numbers.  The partition function  for quantum weighted double Hurwitz numbers can be interpreted in 
   terms of the energy distribution of a quantum Bose gas with vanishing fugacity.
   We compute the leading semiclassical term of the partition function for three  versions
   of the quantum weighted Hurwitz numbers, as well as lower order semiclassical
   corrections. The classical limit $\hbar \ra 0$  is shown to reproduce the
simple single and double Hurwitz numbers studied by Pandharipande and Okounkov \cite{Pa, Ok}.
The KP-Toda $\tau$-function that serves  as generating function for the quantum Hurwitz
numbers is shown to have the $\tau$-function of \cite{Pa, Ok}  as its leading term in the
classical limit, and, with suitable scaling, the same holds for the partition function, the weights and expectations of Hurwitz numbers.
We also compute the zero temperature limit $T \ra 0$ of the partition function  and  quantum weighted Hurwitz numbers.
The KP or Toda $\tau$-function serving  as generating function for the quantum Hurwitz
numbers are shown to give the one for Belyi curves in the zero temperature limit  and, with
suitable scaling, the same holds true for the partition function, the weights and the expectations 
of Hurwitz numbers.

     \end{abstract}

\tableofcontents


\section{Introduction: weighted Hurwitz numbers and their generating functions}
\label{sec:intro}

\subsection{Hurwitz numbers}
\label{subsev:hurwitz}

Multiparametric weighted  Hurwitz numbers  were 
introduced in \cite{GH1, GH2, H1, HO} as generalizations of simple
Hurwitz numbers \cite{Hu1, Hu2, Pa, Ok} and  other special cases 
 \cite{GGN, AC1, AC2, AMMN1, AMMN2, KZ, Z}  previously studied. In general,  parametric families of
 KP or $2D$ Toda $\tau$-functions of {\it hypergeometric type} \cite{KMMM, OrSc} serve as generating 
 functions for the weighted Hurwitz numbers, which appear as coefficients  in an expansion
 over the basis of power sum symmetric functions in an auxiliary set of variables.
  The weights are determined by a parametric family  of weight
 generating functions $G(z, {\bf c})$, with parameters ${\bf c}=(c_1, c_2, \dots)$ that can either be expressed as a formal  sum
 \be
 G(z) = 1 + \sum_{i=1}^\infty g_i z^i
  \label{G_weight_gen_sum}
 \ee
 or an infinite product
 \be
 G(z) = \prod_{i=1}^\infty (1 +z c_i),
 \label{G_weight_gen_prod}
 \ee
 or some limit thereof. Comparing the two formulae, $G(z)$ can be interpreted as the generating
 function for elementary symmetric functions in the variables ${\bf c} =(c_1, c_2, \dots)$.
 \be
 g_i = e_i({\bf c}).
 \ee
 
  Another  parametrization considered in \cite{GH1, GH2, H1, HO},
consists of weight generating functions of the form
\be
 \tilde{G}(z) = \prod_{i=1}^\infty (1 - z c_i)^{-1}.
  \label{tilde_G_weight_gen}
 \ee
 The corresponding power series expansions
\be
 \tilde{G}(z) = 1 + \sum_{i=1}^\infty \tilde{g}_i z^i
 \ee
can similarly be interpreted as defining the complete symmetric functions
 \be
 \tilde{g}_i = h_i({\bf c}).
 \ee

Pure Hurwitz numbers $H(\mu^{(1)}, \cdots , \mu^{(k)})$ may be defined in one of two
 equivalent ways: geometrical or combinatorial. The geometrical definition is:
 \begin{definition}
 For a set of $k$ partitions $(\mu^{(1)}, \cdots , \mu^{(k)})$ of $n$, $H(\mu^{(1)}, \cdots , \mu^{(k)})$
 is the number of distinct $n$-sheeted branched coverings $\Gamma \ra \Pb^1$ of the Riemann sphere
 having $k$ branch points $(p_1, \dots , p_k)$ with ramification profiles $\{\mu^{(i)}\}_{i=1, \dots , k}$,
 divided by the order $\aut(\Gamma)$ of the automorphism group of $\Gamma$.
 \end{definition}
 The combinatorial definition is:
 \begin{definition}  $H(\mu^{(1)}, \cdots , \mu^{(k)})$ is the number of distinct factorization
 of the identity element $\Ib \in S_n$ of the symmetric group as an ordered product 
 \be
 \Ib = h_1, \dots h_k, \quad h_i \in S_n, \quad i=1, \dots , k
\ee
where $h_i$ belongs to the conjugacy class with cycle lengths equal to the parts \
of $\mu^{(i)}$, divided by $n!$.
\end{definition}
The fact that these coincide \cite{GH2, H1} follows from the monodromy representation
of the fundamental group of the sphere minus the branch points mapped into the symmetric group $S_n$ 
Let $\mP_n$ denote the set of integer partitions of $n$ and $p(n)$  its cardinality. The Frobenius-Schur formula \cite{Frob1, Frob2, Sch, LZ} expresses the Hurwitz numbers in terms of the irreducible characters 
of  $S_n$
\be
H(\mu^{(1)}, \dots, \mu^{(k)}) = \sum_{\lambda\in \mP_n} h_\lambda^{k-2} \prod_{i=1}^k  z_{\mu^{(i)}}^{-1} \chi_\lambda(\mu^{(i)}),  
\label{frob_schur}
\ee
where $\chi_\lambda(\mu^)  $ is the irreducible character of the representation
with Young symmetry class $\lambda$ evaluated on the conjugacy class with cycle lengths equal to the parts of $\mu$;
$h_\lambda$ is the product of hook lengths of the Young diagram of partition $\lambda$ and
\be
z_\mu = \prod_{i=1}^{\ell(\mu} m_i(\mu) ! i^{m_i(\mu)}
\ee
is the order of the stabilizer of any element of the conjugacy class  $\mu$,
with $m_i(\mu)$ equal to the number of times $i$ appears as a part of $\mu$. We denote the 
weight of a partition $| \mu|$ , its length $\ell(\mu)$ and define its {\it colength} as
\be
\ell^*(\mu):= |\mu| - \ell(\mu).
\ee

\subsection{Weighted Hurwitz numbers}
\label{subsec:weighted}

Following \cite{GH1, GH2, H1, HO} we define, for each positive integer $d$
and every pair of ramification profiles $(\mu, \nu)$ (i.e. partitions of $n$),
 the weighted double Hurwitz number 
   \be
H^d_G(\mu, \nu) := \sum_{k=0}^\infty \sideset{}{'}\sum_{\substack{\mu^{(1)}, \dots \mu^{(k)} \\ \sum_{i=1}^k \ell^*(\mu^{(i)})= d}}
m_\lambda ({\bf c})H(\mu^{(1)}, \dots, \mu^{(k)}, \mu, \nu) ,
\label{Hd_G}
\ee
where
\be
m_\lambda ({\bf c}) :=
\frac{1}{\abs{\aut(\lambda)}} \sum_{\sigma\in S_k} \sum_{1 \le i_1 < \cdots < i_k}
 c_{i_\sigma(1)}^{\lambda_1} \cdots c_{i_\sigma(k)}^{\lambda_k},
 \label{monomial_sf}
\ee
is the monomial sum symmetric  function \cite{Mac} corresponding to a partition $\lambda$ of weight 
\be
|\lambda|= d = \sum_{i=1}^k \ell^*(\mu^{(i)})
\label{d_def}
\ee
whose parts $\{\lambda_i\}_{i=1, \dots ,k}$ are the colengths $\{\ell^*(\mu^{(i)})\}_{i=1, \dots ,k}$ 
in weakly descending order,
\be
|\aut (\lambda)| := \prod_{i=1}^{\ell(\lambda)} m(\lambda_i)!,
\ee
and $\sum'$ denotes the sum over all $k$-tuples of partitions $(\mu^{(1)}, \dots, \mu^{(k)})$ 
 satisfying condition (\ref{d_def}) other than the cycle type of the identity element.
By the Riemann-Hurwitz formula, the Euler characteristic of the covering surface is
\be
\chi = 2-2g = \ell(\mu) + \ell(\nu)  - d.
\ee

For weight generating functions of the form (\ref{tilde_G_weight_gen}), the weighted double Hurwitz 
number is defined as:
\be
H^d_{\tilde{G}}(\mu, \nu) := \sum_{k=0}^\infty \sideset{}{'}\sum_{\substack{\mu^{(1)}, \dots \mu^{(k)} \\ \sum_{i=1}^k \ell^*(\mu^{(i)})= d}}
f_\lambda ({\bf c})H(\mu^{(1)}, \dots, \mu^{(k)}, \mu, \nu) 
\label{Hd_tilde_G}
\ee
where
\be
f_\lambda ({\bf c}) :=
\frac{(-1)^{\ell^*(\lambda)}}{\abs{\aut(\lambda)}} \sum_{\sigma\in S_k} \sum_{1 \le i_1 \le \cdots \le i_k} 
c_{i_\sigma(1)}^{\lambda_1},  \cdots c_{i_\sigma(k)}^{\lambda_k},
 \label{forgotten_sf}
\ee
is the ``forgotten'' symmetric function \cite{Mac}.

The particular case where all the $\mu_i$'s represent simple branching (i.e. where they are all 2-cycles) was
 studied in \cite{Pa, Ok} and corresponds to the exponential weight generating function
\be
G = \exp, \quad G(z) = e^z =\lim_{k\ra \infty}(1 + {z /k})^k
\label{G_exp}
\ee
The evaluation of the monomial sum symmetric function in this limit is
\be
\label{eq:OkPP}
\lim_{k\ra \infty} m_\lambda\rb{\underbrace{\frac1m \dots, \frac1m}_{k\ \rm{times}}, 0, 0, \cdots} = \delta_{\lambda, (2, (1)^{n-2})},
\ee
so the weight is uniform  on all $k$-tuples $(\mu^{(1)}, \cdots, \mu^{(k)})$ of  partitions corresponding 
to simple branching
\be
\mu^{(i)} = (2, (1)^{n-2})
\label{simple_branching}
\ee
and vanishes on all others.  The corresponding weighted Hurwitz numbers $H^d_{\exp} (\mu, \nu)$  are what we  refer to as the ``classical'' simple (double) Hurwitz numbers.

\subsection{The 2D-Toda $\tau$-function as generating function}
\label{subsec:taufn}

Choosing a small parameter $\beta$, the following double Schur function  expansion defines a $2D$-Toda $\tau $-function of hypergeometric
type \cite{OrSc}  (at the lattice point $N=0$).
\be
   \tau^{(G, \beta)} ({\bf t}, {\bf s})  = \sum_{\lambda} \ r^{(G, \beta)}_\lambda s_\lambda ({\bf t}) s_\lambda ({\bf s}),
   \label{tau_G}
  \ee
  where the coefficients $r^{(G, \beta)}_\lambda$  are defined in terms of the weight generating
  function $G$ by the following {\it content product} formula
  \be
r_\lambda^{(G, \beta)} :=   \prod_{(i,j)\in \lambda} G(\beta( j-i)),
\label{r_G_lambda}
\ee
The same formulae apply {\it mutatis mutandis} with  the replacement $G \ra \tilde{G}$ for the case of the second type of weight generating 
function $\tilde{G}$ defined by (\ref{tilde_G_weight_gen}).

Changing the expansion basis from diagonal products Schur functions to  products $p_\mu({\bf t}) p_\nu({\bf s})$ of
power sum symmetric functions, using the standard  Frobenius character formula \cite{FH, Mac},
\be
s_\lambda = \sum_{\mu, \, \abs{\mu} = \abs{\lambda}} z_\mu^{-1} \chi_\lambda(\mu) p_\mu,
\label{frobenius_character}
\ee
it  follows \cite{GH1, GH2, H1, HO}  that $ \tau^{(G, \beta)} ({\bf t}, {\bf s})) $ is interpretable as a generating function for the weighted double Hurwitz numbers $H^d_G(\mu, \nu)$.
\begin{theorem}{\cite{GH1, GH2, H1, HO}}
\label{tau_H_G_generating function}
The 2D Toda $\tau$-function $\tau^{(G, \beta)}({\bf t}, {\bf s})$
can be expressed as
\be
\tau^{(G, \beta)}({\bf t}, {\bf s}) =\sum_{d=0}^\infty \beta^d \sum_{\substack{\mu, \nu \\ |\mu|=|\nu|}} H^d_G(\mu, \nu) p_\mu({\bf t}) p_\nu({\bf s}).
\label{tau_H_G}
\ee
and the same formula holds under the replacement $G \ra \tilde{G}$.
\end{theorem}
The case of the exponential weight generating function (\ref{G_exp}) gives the following
content product coefficient in the $\tau$-function expansion (\ref{tau_G})
\be
r^{(\exp, \beta)}_\lambda = e^{{\beta \over 2} \sum_{i=1}^{\ell(\lambda)}\lambda_i(\lambda_i - 2i +1)},
\ee
as in \cite{Ok}, and the generating function expansion (\ref{tau_H_G}) becomes
\be
\tau^{(\exp, \beta)}({\bf t}, {\bf s}) =\sum_{k=0}^\infty {\beta^d \over d!} \sum_{\substack{\mu, \nu \\ |\mu|=|\nu|}} H^d_{\exp} (\mu, \nu)
p_\mu({\bf t}) p_\nu({\bf s}), 
\ee
where
\be
H^d_{\exp} (\mu, \nu):= H(\underbrace{(2, (1)^{n-2}) , \dots, (2, (1)^{n-2})}_{d\  {\rm times}}).
\label{simple_double_Hurwitz}
\ee

\subsection{Quantum Hurwitz numbers: Relation to Bose gases}
\label{subsec:bosegases}

\subsubsection{Quantum Hurwitz numbers}

A special case of weighted Hurwitz numbers consists of {\em pure quantum Hurwitz numbers} \cite{GH2, H2},
  which are obtained by choosing the parameters $c_i$ as
\be
c_i = q^i, \quad i=1, 2 \dots 
\label{c_iq_i_prime}
\ee
where $q$ is a real parameter between $0$ and $1$.  The justification for naming these {\it quantum Hurwitz numbers}
is the relation, under suitable identification of parameters, of the energy distribution
of the various branched configurations to that for a quantum  Bose gas with linear energy spectrum,
The classical limit will be seen to coincide with the case of the Dirac measure supported on the space of simple branchings of
type (\ref{simple_branching}), as studied by Okounkov and Pandharipande \cite{Ok, Pa}.

The corresponding weight generating function is 
\bea
G(z) = E'(q,z) &\& \deq \prod_{i=1}^\infty (1+ q^i z) = (-zq; q)_\infty:=1 + \sum_{i=0}^\infty E'_i(q) z^i,
\\
E'_i(q) &\& \deq \frac{q^{\frac{1}{2}i(i+1)}}{\prod_{j=1}^i (1-q^j)} =  \frac{q^{\frac{1}{2}i(i+1)}}{(q;q)_{i-1}} , \quad i \ge 1,
\label{E_prime_qz__def}
\eea
where
\be
(z;q)_k := \prod_{j=0}^{k-1}((1 - z q^j), \quad (z; q)_\infty := \prod_{j=0}^{\infty}(1 - z q^j)
\ee
is the quantum Pochhammer symbol. This is related to the quantum dilogarithm function by
\be
(1+z) E'(q, z) = e^{-{\Li_2(q, -z)\over 1-q}}, \quad \Li_2(q, z) \deq (1-q) \sum_{k=1}^\infty \frac{z^k}{k (1- q^k)}.
\ee
We thus have 
\be
e_\lambda({\bf c}) = :E'_\lambda(q) = \prod_{i=1}^{\ell(\lambda)}\frac{q^{\frac{1}{2}\lambda_i(\lambda_i +1)}}{\prod_{j=1}^{\lambda_i} (1-q^j)} = \prod_{i=1}^{\ell(\lambda)}\frac{q^{\frac{1}{2}\lambda_i(\lambda_i +1)}}{(q;q)_{\lambda_i-1}}  .
\ee
The content product  coefficients entering in the $\tau$-function (\ref{tau_G}) for this case are
\bea
r^{(E'(q), \beta)}_j &\&= \prod_{k=1}^\infty (1+ q^k \beta j) = (-q\beta j; q)_{\infty} , \\
r^{(E'(q), \beta)}_\lambda(z) &\&= \prod_{k=1}^\infty \prod_{(i,j)\in \lambda} (1+ q^k \beta (j-i)) 
 = \prod_{(i,j)\in \lambda} (-q\beta(j-i); q)_\infty.
 \eea

Making the substitutions (\ref{c_iq_i_prime}), the weights entering in (\ref{Hd_G}) evaluate to
\bea
\label{eq:WePrime}
W_{E'(q)} (\mu^{(1)}, \dots, \mu^{(k)}) &\& := m_\lambda (q, q^2, \dots ) \cr
&\& =  {1\over  |\aut(\lambda)|} \sum_{\sigma\in S_k} \frac{1}{
(q^{-\ell^*(\mu^{(\sigma(1))})} -1) \cdots (q^{-\ell^*(\mu^{(\sigma(1))})} \cdots q^{-\ell^*(\mu^{(\sigma(k))})}-1)}, \cr
&\&
\label{W_Eprime_q}
\eea
The (unnormalized) weighted Hurwitz numbers  therefore become
  \be
H^d_{E'(q)}(\mu, \nu) := \sum_{k=0}^\infty \sideset{}{'}\sum_{\substack{\mu^{(1)}, \dots \mu^{(k)} \\ \sum_{i=1}^k \ell^*(\mu^{(i)})= d}}
W_{E'(q)} (\mu^{(1)}, \dots, \mu^{(k)}) H(\mu^{(1)}, \dots, \mu^{(k)}, \mu, \nu). 
\label{Hd_E_prime_q}
\ee

Choosing $G= E'(q)$ in eqs.~(\ref{tau_G}), (\ref{tau_H_G}), we obtain
\bea
   \tau^{(E'(q), \beta)} ({\bf t}, {\bf s})  &\&= \sum_{\lambda} \ r^{(E'(q), \beta)}_\lambda(z)s_\lambda ({\bf t}) s_\lambda ({\bf s}) \cr
   &\& =\sum_{d=0}^\infty \beta^d \sum_{\substack{\mu, \nu \\ |\mu|=|\nu|}} H^d_{E'(q)}(\mu, \nu) p_\mu({\bf t}) p_\nu({\bf s}).
   \label{tau_Eprimeq}
\eea
as generating function for  simple quantum Hurwitz numbers.

\subsubsection{Relation to Bose gas}

The parameter $q$ may be interpreted as $q = e^{-\epsilon}$, for a small positive parameter $\epsilon$ identified as
\be
\epsilon = { \hbar \omega_0 \over k_B T} ,  \quad T= \text{temperature}, \quad k_B =\text{Boltzman constant},
\ee
where $\hbar \omega_0$ is interpreted as a ground state energy (i.e., no branching), 
while the higher levels $\epsilon (\mu)$ are integer multiples
proportional to the colength of the partition representing the ramification type of a branch point; i.e., the
degree of degeneration of the sheets
\be
\epsilon (\mu) = \ell^*(\mu) \epsilon.
\ee
The weight $W_{E'(q)} (\mu^{(1)}, \dots, \mu^{(k)})$ for a branching configuration of type $(\mu^{(1)}, \dotsm \mu^{(k)}, \mu, \nu)$ is then
\be
W_{E'(q)} (\mu^{(1)}, \dots, \mu^{(k)})  =  {1\over  |\aut(\lambda)|} \sum_{\sigma\in S_k} \frac{1}{
(e^{(\epsilon(\mu^{(\sigma(1)})} -1) \cdots (e^{(\sum_{i=1}^k\epsilon(\mu^{(\sigma(i)})}-1) }.
\label{eq:W_Eprime_hbar}
\ee

The weight $W_{E'(q)} (\mu^{(1)}, \dots, \mu^{(k)})$ defined in  (\ref{W_Eprime_q}) may thus be interpreted 
in terms of the distribution of a quantum bose gas with vanishing fugacity
\be
n_{\epsilon(\mu)}= {1 \over e^{\epsilon(\mu)} -1},
\ee
assuming that the energy for $k$ branch points $(\mu^{(1)}, \dots, \mu^{(k)})$ is the sum of that for each
\be
\epsilon(\mu^{(1)}, \dots, \mu^{(k)}) = \sum_{i=1}^k \epsilon(\mu^{(i)})
\ee
and that the total weight for configurations with up to $k$ weighted  branch
points (plus two unweighted ones, with profiles $(\mu, \nu)$  is the sum over the products of those for $i=1, 2, \dots k$ branch points.

\subsubsection{Other variants of quantum Hurwitz numbers}

Another variant  on the weight generating function for quantum Hurwitz numbers consists of choosing the 
parameters ${\bf c} = (c_1, c_2, \dots)$ in (\ref{G_weight_gen_prod}) to be
\be
c_i := q^{i-1},
\label{c_iq_i}
\ee
which gives
\be
G(z) =E(q, z) \deq (-qz;q)_\infty 
\label{Eqz_def}
\ee
This is related to the quantum dilogarithm function by
\be
E(q, z) = e^{-{\Li_2(q, -z)\over 1-q}}.
\ee
The content product  coefficients entering in the $\tau$-function (\ref{tau_G}) for this case are
\bea
r^{E(q)}_j(z) &\&= \prod_{k=0}^\infty (1+ q^k z j) = (-zj; q)_\infty, \\
r^{E(q)}_\lambda(z) &\&= \prod_{k=0}^\infty \prod_{(i,j)\in \lambda} (1+ q^k z (j-i)) 
 = \prod_{(i,j)\in \lambda} (-z(j-i); q)_\infty 
\eea

The weights entering in (\ref{Hd_G}) evaluate to
\begin{align}
	\label{eq:We}
W_{E(q)} (\mu^{(1)}, \dots, \mu^{(k)}) & \deq {1\over |\aut(\lambda)|}
\sum_{\sigma\in S_k} \sum_{0 \le i_1 < \cdots < i_k}^\infty q^{i_1 \ell^*(\mu^{(\sigma(1))})} \cdots q^{i_k \ell^*(\mu^{(\sigma(k))})} \\
& = {1\over  |\aut(\lambda)|}\sum_{\sigma\in S_k} \frac{q^{(k-1) \ell^*(\mu^{(\sigma(1))})} \cdots q^{\ell^*(\mu^{(\sigma(k-1))})}}{
(1- q^{\ell^*(\mu^{(\sigma(1))})}) \cdots (1- q^{\ell^*(\mu^{(\sigma(1))})} \cdots q^{\ell^*(\mu^{(\sigma(k))})})},
\label{W_E_q}
\end{align}
and the  weighted Hurwitz numbers  therefore become
  \be
H^d_{E(q)}(\mu, \nu) := \sum_{k=0}^\infty \sideset{}{'}\sum_{\substack{\mu^{(1)}, \dots \mu^{(k)} \\ \sum_{i=1}^k \ell^*(\mu^{(i)})= d}}
W_{E(q)} (\mu^{(1)}, \dots, \mu^{(k)}) H(\mu^{(1)}, \dots, \mu^{(k)}, \mu, \nu). 
\label{Hd_E_q}
\ee

A third variant on the  weight generating function for quantum Hurwitz numbers consists 
of choosing it of the form (\ref{tilde_G_weight_gen}) with parameters ${\bf c} = (c_1, c_2, \dots)$ again chosen as in (\ref{c_iq_i}).
This gives
\bea
\tilde{G}(z) =H(q,z)&\&\deq \prod_{k=0}^\infty (1-q^k z)^{-1} ={1\over (-z;q)_\infty}= e^{\Li_2(q, z)} = \sum_{i=0}^\infty H_i(q)z^i,
\label{GHq} \\
r^{H(q)}_j(z) &\&= \prod_{k=0}^\infty (1- q^k z j)^{-1} = {1\over (-z;q)_\infty}, \\
r^{H(q)}_\lambda(z) &\&= \prod_{k=0}^\infty \prod_{(i,j)\in \lambda} (1- q^kz (j-i)) ^{-1} 
 = \prod_{(i,j)\in \lambda} {1\over (-z(j-i);q)_\infty}.
\label{rHq}
\eea

The weights entering in (\ref{Hd_G}) then evaluate to
\begin{align}
\label{eq:Wh}
W_{H(q)} (\mu^{(1)}, \dots, \mu^{(k)}) & \deq
{(-1)^{\ell^*(\lambda)}\over   |\aut(\lambda)|}\sum_{\sigma\in S_k} \sum_{0 \le i_1 \le \cdots \le i_k}^\infty q^{i_1 \ell^*(\mu^{(\sigma(1))})} \cdots q^{i_k \ell^*(\mu^{(\sigma(k))})} \\
&= {(-1)^{\ell^*(\lambda)}  \over  |\aut(\lambda)|}\sum_{\sigma\in S_k} \frac{1}{
(1- q^{\ell^*(\mu^{(\sigma(1))})}) \cdots (1- q^{\ell^*(\mu^{(\sigma(1))})} \cdots q^{\ell^*(\mu^{(\sigma(k))})})}
\label{W_H_q}
\end{align}
and the weighted Hurwitz numbers  become
  \be
H^d_{H(q)}(\mu, \nu) := \sum_{k=0}^\infty \sideset{}{'}\sum_{\substack{\mu^{(1)}, \dots \mu^{(k)} \\ \sum_{i=1}^k \ell^*(\mu^{(i)})= d}}
W_{H(q)} (\mu^{(1)}, \dots, \mu^{(k)}) H(\mu^{(1)}, \dots, \mu^{(k)}, \mu, \nu). 
\label{Hd_H_q}
\ee

\subsection{Asymptotics for quantum Hurwitz numbers}
\label{subsec:asymptotics_hurwitz}

\subsubsection{Classical asymptotics for quantum Hurwitz numbers}
\label{classical_asymptotics}

With the identification $q= e^{-\epsilon}$ and $\epsilon := { \hbar \omega_0 \over k_B T}$
  viewed as a small positive number, taking the limit $\epsilon \ra 0^+$
of the scaled quantum dilogarithm function $Li_2(q, \epsilon z)$ gives
\be
\lim_{\epsilon \ra 0^+} {Li_2(e^{-\epsilon}, \epsilon z) \over 1- e^{-\epsilon}}, =z.
\ee
It follows that all three generating functions $E(q,z), E'(q,z)$ and $H(q,z)$ have as scaled limits
the generating function for the Okounkov-Pandharipande simple (single and double) Hurwitz numbers
\be
\lim_{\epsilon \ra 0^+} E(q,\epsilon z)= \lim_{\epsilon \ra 0^+} E'(q,\epsilon z)= \lim_{\epsilon \ra 0^+} H(q,\epsilon z) = e^z,
\ee

The corresponding scaled limit of the generating $\tau$-functions for all three versions of quantum 
weighted Hurwitz numbers therefore coincides with the generating function for simple Hurwitz
numbers considered in \cite{Pa, Ok}
\be
\lim_{\epsilon \ra 0^+} \tau^{(E(q), \epsilon \beta)}({\bf t}, {\bf s}) = \lim_{\epsilon \ra 0^+} \tau^{(E'(q), \epsilon \beta)}({\bf t}, {\bf s})
= \lim_{\epsilon \ra 0^+} \tau^{H(q), \epsilon \beta)}({\bf t}, {\bf s}) = \tau^{(\exp,  \beta)}({\bf t}, {\bf s}). 
\ee
Equivalently, this implies the limit
\be
\label{eq:OkPaLimit}
\lim_{\epsilon \ra 0^+} \epsilon^d H^d_{E'(q= e^{\epsilon})} = H^d_{\exp} (\mu, \nu)
\ee
(cf. Theorem \ref{Hd_E_prime_semiclassical} and Remark \ref{rmk:OkPaLimit}).

\subsubsection{Zero temperature asymptotics for quantum Hurwitz numbers}
\label{zero_temp_asymptotics}

The zero temperature limit $T \ra 0$, on the other hand, corresponds to $q \ra  0^+$. The suitably scaled limit of the
generating function $E'(q,z)$ is
\be
\lim_{\epsilon \ra \infty^+} E'(q=e^{-\epsilon}, z e^{\epsilon} )= 1+z,
\ee
and the corresponding limit of the $\tau$-function is
\be
 \lim_{\epsilon \ra \infty^+} \tau^{E'(q= e^{-\epsilon},  e^\epsilon \beta)}({\bf t}, {\bf s})= \tau^{(E,  \beta)}({\bf t}, {\bf s}). 
\ee
where the weight generating function  $E$ is
\be
E(z) = 1+z.
\ee
This is the generating function for uniformly weighted Hurwitz numbers supported on curves with just
three branch points $(\mu^{(1)}, \mu, \nu)$ \cite{GH2, H1}; i.e. those with $k=1$, sometimes referred 
to as {\it Belyi curves} \cite{AC1, KZ, Z}. (cf. Theorem \ref{thm:DownstairsZeroTemp} and Remark \ref{belyi_limit_E}).


\section{Probabilistic approach to quantum Hurwitz numbers}

Since $W_{E'(q)}\rb{\mu^{(1)},\ldots,\mu^{(k)}}$ is real, positive and normalizable, we can interpret $H^d_{E'(q)}$ in terms of an expectation value.
For $k\in\set{1, \dots , d}$ consider the (finite) set of $k$-tuples
\begin{align}
		\label{eq:defM}
		\fM_{d,k}^{(n)} & = \set{\rb{\mu^{(1)}, \ldots, \mu^{(k)}} \in \rb{\mP_n}^k \colon \sum_{j=1}^k \ell^\ast \rb{\mu^{(j)}} =d } 
		\intertext{and their disjoint union}  
		\fM_d^{(n)}&=\coprod_{k=1}^{d} \fM^{(n)}_{d,k}.
\end{align}
Define a measure $\Pmp$ on $\fM^{(n)}_{d}$ by
\begin{align}
	\label{eq:defTheta}
	\Pmp \rb{ \rb{\mu^{(1)}, \ldots, \mu^{(k)}} } & = \frac1{\Pfp} W_{E'(q)}\rb{\mu^{(1)},\ldots,\mu^{(k)} },
	\intertext{where the \emph{partition function} $\Pfp$ is defined so that $\Pmp$ is a probability measure; that is,}
	\label{eq:defZt}
	\Pfp &= \sum_{k=1}^d \sum_{\fM^{(n)}_{d,k}} W_{E'(q)}\rb{\mu^{(1)},\ldots,\mu^{(k)} }.
	\intertext{We then have the expectation value}
	\label{eq:wHexp}
	\ip{H\rb{\cdot,\ldots,\cdot,\mu,\nu}}_{\Pmp} & = \frac1{\Pfp} H_{E'(q)}^d(\mu,\nu),
\end{align}
where $\ip{\cdot}_{\Pmp}$ denotes integration with respect to the measure $\Pmeas dnq$.

\begin{definition}
	For $n,d\in\Zbb_{>0}$ define the function $\lambdamap\colon \fM_d^{(n)}\longrightarrow \mP_d$ as follows: 
\begin{align}
	\lambdamap & \colon \rb{\mu^{(1)},\ldots,\mu^{(k)}}\longmapsto \lambda
	\intertext{where $\lambda$ is the unique partition of $d$ such that}
	\set{\lambda_1,\ldots,\lambda_k} & = \set{ \ell^\ast\rb{\mu^{1} },\ldots, \ell^\ast\rb{\mu^{(k)} }  }.
\end{align}
\end{definition}

The weight of the partition $\lambdamap\rb{\mu^{(1)},\ldots,\mu^{(k)}}$ is thus the sum $d$ of colengths of 
the partitions $\{\mu^{(1)},\ldots,\mu^{(k)}\}$

\be
d = \sum_{i=1}^k \ell^*(\mu^{(i)}.
\ee
Letting $\mP_{n,k}$ denote the set of partitions of $n$ with $k$ parts, 
the image of $\fM_{d,k}^{(n)}$ under $\lambdamap$ is thus $\mP_{d,k}$.

Since $W_{E'(q)}\rb{\mu^{(1)},\ldots,\mu^{(k)}}$ depends on the partitions $\{\mu^{(1)},\ldots,\mu^{(k)}\}$ 
only through their colengths, it makes sense to consider the push-forward
\begin{align}
	\label{eq:PushForward}
	\pmppre & = \rb{\lambdamap}_\ast\Pmp
	\intertext{of $\Pmp$ under $\lambdamap$ (as a measure on $\mP_d$). Let $p(n,k):=\abs{\mP_{n,k}}$ denote the cardinality of $\mP_{n,k}$ and observe that, for any $\lambda\in \mP_d$,}
	\label{eq:Cardinality}
	\abs{\rb{\lambdamap}^{-1}(\lambda)}& = \prod_{j=1}^{\ell(\lambda)} p\rb{n,n-\lambda_j}.
	\intertext{Therefore}
	\label{eq:defXiPrel}
	\pmppre(\lambda) & = \frac1{\Pfp}\, \rb{\prod_{j=1}^{\ell(\lambda)} p\rb{n,n-\lambda_j} } w_{E'(q)}(\lambda)
\end{align}
where $w_{E'(q)}$ is defined as the weight function $w_{E'(q)}\colon \mP_d\longrightarrow [0,\infty)$ satisfying
\begin{align}
	\label{eq:defw}
	w_{E'(q)}(\lambda)  & = \frac{\Phi_{E'(q)}(\lambda_1,\ldots,\lambda_{\ell(\lambda)})}{\abs{\aut(\lambda)}}
	\intertext{with $\Phi_{E'(q)}\colon\coprod_{m\in\Nbb} \Rbb^m\longrightarrow \rr$ defined by}
	\label{eq:defPhi}
	\Phi_{E'(q)}\rb{x_1,\ldots,x_m} & = \sum_{\sigma \in S_m} \prod_{j=1}^m \rb{ q^{-\sum_{i=1}^j x_{\sigma(i)} } -1 }^{-1}.
\end{align}

\begin{lemma}
	\label{lem:countColengths}
	For any $n,\ell\in\bN$ with $n\geq 2\ell$ we have
	\begin{align}
		p(n,n-\ell)=p(\ell).
	\end{align}
\end{lemma}

\par\noindent The proof of this lemma is given in Section \ref{sec:proofs}. From now on we always assume that $n\geq 2d$. We  also denote
\begin{align}
	p(\lambda)=\prod_{j=1}^{\ell(\lambda)} p\rb{\lambda_j}.
\end{align}
From the above discussion and Lemma \ref{lem:countColengths} we have the following result. For $d\in\Zbb_{> 0}$ and $q\in (0,1)$ let
\begin{align}
	\label{eq:defpfn}
	\pfp&:= \sum_{\lambda\in\mP_d} p(\lambda)\, w_{E'(q)}(\lambda)
	\intertext{and define a probability measure on $\mP_d$ by}
	\label{eq:defXi}
	\pmp(\lambda) &:= \frac1{\pfp}\, p\rb{\lambda} w_{E'(q)}(\lambda)  \quad\quad\forall\, \lambda\in\mP_d.
\end{align}

\begin{proposition}
	Let $n,d\in\Zbb_{>0}$ with $n\geq 2d$. Then
	\begin{enumerate}
		\item The partition function $\Pfp$ does not depend on $n$; 
		\begin{align}
			\Pfp & = \pfp
		\end{align}
		\item The probability measure $\pmppre$ does not depend on $n$: for any $\lambda\in\mP_d$,
		\begin{align}
			\label{eq:pmppre}
			\pmppre (\lambda) &= \pmp (\lambda)
		\end{align}
	\end{enumerate}
\end{proposition}

We conclude this section by explaining how this extends to the other two quantum weight generating functions $E(q)$ and $H(q)$.

\begin{definition}
	Define probability measures $\Pme$ and $\Pmh$ on $\fM_d^{(n)}$ as in \eqref{eq:defTheta} and \eqref{eq:defZt}, replacing $W_{E'(q)}$ by $W_{E(q)}$ and $W_{H(q)}$ respectively, whenever it occurs.
\end{definition}

Equations \eqref{eq:wHexp}, \eqref{eq:PushForward}--\eqref{eq:defPhi} and \eqref{eq:defpfn}--\eqref{eq:pmppre}  apply 
{\em mutatis mutandis}, replacing $E'(q)$ by $E(q)$ and $H(q)$ respectively.


\section{Semiclassical limits and asymptotic expansion}

\label{sec:scl}

In this section we state our asymptotic results for the semiclassical limit $q\longrightarrow 1^-$. All proofs are given in the Section \ref{sec:proofs}. 


\subsection{Classical limit}

 We begin by stating the classical limits (i.e., the leading term).
 \begin{definition}
	The \emph{Dirac measure} $\delta_x$ at $x\in S$ on a measurable space $(S,\Sigma)$ is defined by
	\begin{align}
		\label{eq:defDirac}
		\delta_x(A) & = \begin{cases}
			1\quad&\text{if } x\in A\\
			0& \text{otherwise}
		\end{cases}
	\end{align}
	for all $A\in\Sigma$.
 \end{definition}

Recall that $m_i(\lambda)$ denotes the number of parts of $\lambda$ equal to $i$
We can then identify the  partition alternatively as
\be
\lambda  = \rb{1^{m_1(\lambda)},2^{m_2(\lambda)},\ldots}.
\ee
We also use the following notation for partitions with at most two different part lengths $( \ell \in \Zb^+, 1)$, 
such that 
\be
m_\ell(\lambda) = m, \quad m_1(\lambda)= n-m
\ee
\begin{definition}
\label{part_notat}
 We denote such a partition as
\be
\label{eq:specialPart}
\bm{\ell}^m_n:=\rb{1^{n-m},\ell^m}.
\ee
	When the weight of the partition is clear from context we simply write $\bm{\ell}^m: = \bm{\ell}^m_n$. When $m=1$ we write $\bm{\ell}_n:= \bm{\ell}^1_n$, or simply $\bm{\ell}$. 
	\end{definition}

\begin{figure}[ht]

	\begin{center}
	\begin{tikzpicture}[>=latex,scale=0.5]

		\draw (0,0)--(1,0)--(1,1)--(0,1)--(0,0);
		
		\foreach \y in {1,2,3,4}{
			\draw (0,\y) -- +(0,1) -- +(1,1) -- +(1,0) -- +(0,0);
		}
		
		\draw [decorate,decoration={brace,amplitude=10pt},xshift=-4pt,yshift=0pt]
		(0,5) -- (0,8) node [black,midway,xshift=-0.6cm] 
		{\footnotesize $m$};
		
		\draw [decorate,decoration={brace,amplitude=10pt},xshift=0pt,yshift=4pt]
		(0,8) -- (5,8) node [black,midway,yshift=0.6cm] 
		{\footnotesize $\ell$};
		
		\foreach \y in {5,6,7}{
			\foreach \x in {0,1,2,3,4}{
				\draw (\x,\y) -- +(0,1) -- +(1,1) -- +(1,0) -- +(0,0);
			}
			
		}

		\foreach \y in {0,1,2,3,4,5,6,7}{
			\draw (10,\y) -- +(0,1) -- +(1,1) -- +(1,0) -- +(0,0);
		}

		\foreach \x in {11,12,13,14,15,16}{
			\draw (\x,7) -- +(0,1) -- +(1,1) -- +(1,0) -- +(0,0);
		}

		\draw [decorate,decoration={brace,amplitude=10pt},xshift=0pt,yshift=4pt]
		(10,8) -- (17,8) node [black,midway,yshift=0.6cm] 
		{\footnotesize $\ell$};
													
	\end{tikzpicture}
	\end{center}
	\footnotesize{\caption{The partitions $\bm{\ell}^m_n=\bm{5}^3_{20}$ (left) and $\bm\ell=\bm{7}$ (with $m=1$ and $n=14$ suppressed from the notation, right)}}
\label{fig:specialPart}
\end{figure}
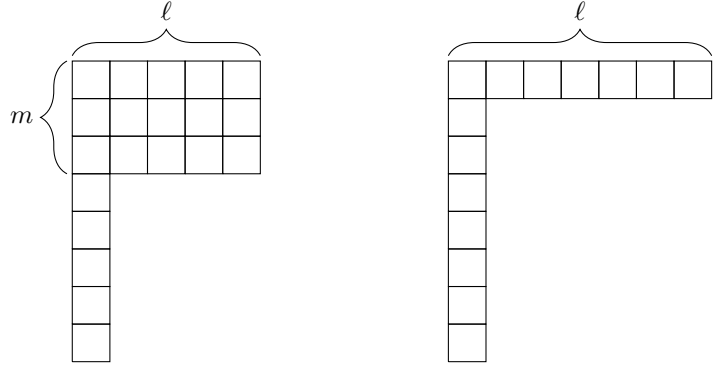

\begin{theorem}
	\label{thm:Downstairs}
	Let $d\in\Zbb_{>0}$. As $q\longrightarrow 1^-$, each of the sequence of measures $\rb{\pmp}_{q<1}$, $\rb{\pme}_{q<1}$ and $\rb{\pmh}_{q<1}$ on $\mP_d$ converges weakly to the Dirac measure $\delta_{(1^d)}$ at $(1^d)\in\mP_d$.
\end{theorem}

By the discussion in Section 2 this translates to a convergence result on $\fM^{(n)}_d$:

\begin{corollary}
	\label{cor:Upstairs}
	 If $d\geq 2n$ then each of the sequence of measures $\Pmp$, $\Pme$ and $\Pmh$ on $\fM_d^{(n)}$ converges weakly, as $q\longrightarrow 1^-$, to the Dirac measure at $(\underbrace{\bm{2},\ldots,\bm{2} }_{d\text{ terms}})$ (in the notation of \eqref{eq:specialPart})
\end{corollary}

\begin{remark}
	Observe that, from \eqref{eq:OkPP}, the limiting measure in Corollary \ref{cor:Upstairs} corresponds 
to the Okounkov-Pandharipande measure \cite{Ok, Pa} .
\end{remark}


\subsection{Semiclassical corrections}

We now turn to the next order term in the semiclassical asymptotics.  Throughout we set $q=e^{-\ep}$ and let $\ep\longrightarrow 0^+$. We begin by giving the asymptotic expansion for each weight. For any $\lambda\in\mP_d$ define 
\begin{align}
	w_0(\lambda) & = \sum_{\sigma\in S_{\ell(\lambda)}}  \frac1{ \prod_{j=1}^{\ell(\lambda)} \sum_{i=1}^j \lambda_{\sigma(i)} }\\
	w_1(\lambda) & = \frac12 \sum_{\sigma\in S_{\ell(\lambda)}}\sum_{r=1}^{\ell(\lambda)} \frac{\sum_{i=1}^r \lambda_{\sigma(i)} }{\prod_{j=1}^{\ell(\lambda)} \sum_{i=1}^j \lambda_{\sigma(i)} }
\end{align}

\begin{theorem}
	\label{prop:weightSC}
	For any $\lambda\in\mP_d$ we have
	\begin{align}
		\ep^{-\ell(\lambda)} w_{E'(e^{-\ep})}(\lambda) & = w_0(\lambda)   + \ep w_1(\lambda) \,      +O\rb{\ep^{2}}\\
		\ep^{-\ell(\lambda)} w_{E(e^{-\ep})}(\lambda) & = w_0(\lambda) + \ep\rb{w_1(\lambda) - dw_0(\lambda)}+O\rb{\ep^{2}}\\
		\ep^{-\ell(\lambda)} w_{H(e^{-\ep})}(\lambda) & = w_0(\lambda) + \ep\rb{w_1(\lambda) - \frac{\ell(\lambda)(\ell(\lambda)+1)}2\, d w_0(\lambda)} +O\rb{\ep^{2}}.
	\end{align}
\end{theorem}

From this result one can deduce the following semiclassical expansion for the partition function:

\begin{theorem}
	\label{cor:pfSC}
	For $d\in\Zbb_{> 0}$ and $q=e^{-\ep}$ we have
	\begin{align}
		\ep^d\, \pf d{E'(e^{-\ep})}	&= \frac{1}{d!} +\ep\, \frac{3-d}{4(d-1)!} + O\rb{\ep^{2}},\\
		\ep^d\, \pf d{E(e^{-\ep})}	&= \frac{1}{d!} +\ep\, \frac{5+d}{4(d-1)!} + O\rb{\ep^{2}},\\
		\ep^d\, \pf d{H(e^{-\ep})}	&= \frac{1}{d!} +\ep\, \frac{d+1}{(d-1)!} + O\rb{\ep^{2}}.
	\end{align}
\end{theorem}

We also obtain a semiclassical convergence result for the weighted Hurwitz numbers. (Recall our notation for partitions from Definition \ref{part_notat}.)

\begin{theorem}
\label{Hd_E_prime_semiclassical}
	For any $\mu,\nu\in\mP_n$ and $G\in\set{E',E,H}$ we have
	\label{thm:WeightedHurwitz}
	\begin{align}
		\ep^d\ H_{G(q)}^{d}(\mu,\nu) & = \frac{1}{d!}\, H(\underbrace{\bm{2},\ldots,\bm{2}}_{d\text{ times}},\mu,\nu) \\
		&\quad + \ep \left[ \gamma(1) H(\underbrace{\bm{2},\ldots,\bm{2}}_{d-1\text{ times}},\bm{3},\mu,\nu)) + \gamma(1) H(\underbrace{\bm{2},\ldots,\bm{2}}_{d-1\text{ times}},\bm{2^2},\mu,\nu)) \right.\\
		 &\quad\quad\left.+\, \gamma_{G(q)}(2) H(\underbrace{\bm{2},\ldots,\bm{2}}_{d\text{ times}},\mu,\nu)\right] + O\rb{\ep^{2}}.
	\end{align}
	where $\gamma(1)= \frac1{(d-1)!}$ and
	\begin{align}
		\gamma_{E'(q)}(2) = -\frac{d+1}{4(d-1)!},\quad\quad \gamma_{E(q)}(2) =-\frac{3-d}{4(d-1)!},\quad\quad \gamma_{H(q)}(2) = \frac{d+1}{2(d-1)!}.
	\end{align}
\end{theorem}

\begin{remark}
	\label{rmk:OkPaLimit}
	In particular, for $G\in\set{E',E,H}$, we have
	\begin{align}
		\lim_{\ep\to 0} H_{G(e^{-\ep})}\rb{\mu,\nu} & = \frac1{d!}\, H_{\exp}(\mu,\nu),
	\end{align}
	which includes \eqref{eq:OkPaLimit} for $G=E'$.
\end{remark}


\section{Zero-temperature limit and asymptotic expansion}

\label{sec:ztl}

Recalling that the parameter $q$ is interpreted as 
\be
	q  = e^{-{ \hbar\omega_0 \over k_B T}}
\ee
for some ground state energy $E_0=\hbar\omega_0$, the zero temperature limit $T\longrightarrow 0^+$ corresponds to $q\longrightarrow 0^+$. In this section we state our asymptotic results in this limit. All proofs are given in the following section. 
We further assume throughout that $d\geq 2$ and $n\geq 2d$.


\subsection{Zero-temperature limit: leading term}

 We begin by stating the leading order zero temperature limit.

\begin{theorem}
	\label{thm:DownstairsZeroTemp}
	Let $d\in\Zbb_{>0}$. As $q\longrightarrow 0^+$, the  $1$-parameter family of measures $\rb{\pmp}$ on $\mP_d$ converges weakly to the Dirac measure $\delta_{(d)}$ at $(d)\in\mP_d$.
\end{theorem}

\ 

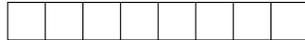
\begin{figure}[ht]

	\begin{center}
	\begin{tikzpicture}[>=latex,scale=0.5]

		\draw (0,0)--(1,0)--(2,0)--(3,0)--(4,0) -- (5,0) -- (6,0) --(7,0) --(8,0) -- (8,1)-- (7,1)-- (6,1) -- (5,1)-- (4,1) -- (3,1) --(2,1) -- (1,1) -- (0,1) -- (0,0) ;
		
		\draw (1,0) --(1,1);
		\draw (2,0) -- (2,1);
		\draw (3,0) -- (3,1);
		\draw (4,0) -- (4,1);
		\draw (5,0) -- (5,1);
		\draw (6,0) -- (6,1);
		\draw (7,0) -- (7,1);
		\draw (8,0) -- (8,1);
			
	\end{tikzpicture}
	\end{center}
	\footnotesize{\caption{The partition $(d)$ on which $\pmp$ concentrates asymptotically in the zero-temperature limit}}
\label{fig:specialPartZeroTemp}
\end{figure}

\

\par\noindent By the discussion in Section 2 this translates to a convergence result on $\fM^{(n)}_d$:

\begin{corollary}
	\label{cor:UpstairsZeroTemp}
	 If $n\geq 2d$ then the measure $\Pmp$ on $\fM_d^{(n)}$ converges weakly, as $q\longrightarrow 0^+$, to the uniform measure $\nu$ on $\fM^{(n)}_{d,1}$, the set of single partitions $\mu^{(1)}$ of $n$ with colength $d$. That is,
	 \begin{align}
	 	\nu(A) & = \frac{\abs{A\cap\fM^{(n)}_{d,1} }}{\abs{\fM^{(n)}_{d,1}}}.
	 \end{align}
\end{corollary}

\begin{remark}
\label{belyi_limit_E}
	Note that this leading term contribution has support on the horizontal partition $\lambda = (d)$.
	This corresponds to branched coverings with just $k=1$ weighted branch point,
	with ramification profile $\mu^{(1)}$ of colength
	\be
	\ell^*(\mu^{(1)}) =d,
	\ee
	plus two unweighted ones, with ramification profiles $(\mu, \nu)$. This means three branch
	points in total, with profiles $(\mu^{(1)}, \mu, \nu)$, where the first branched point, with profile $\mu^{(1)}$,
	has uniform measure on the space of branch points with colength $d$.
	The weighted enumeration of such branched coverings, sometimes known as Belyi curves 
	is known to be equivalent to the enumeration of Grothendieck's {\it Dessins d'Enfants} \cite{AC1, KZ, Z} .
	 They are also determined by a  generating function  $\tau^{(E, \beta)}({\bf t}, {\bf s})$, that has been studied in 
   \cite{GH2, H1}, in which the weight generating function $G=E$ is chosen to be simply
	 \be
   E(z) := 1+z.
   \label{E_weight_generating}
   \ee
   
	\end{remark}


\subsection{Higher-order corrections}

We now consider higher order terms in the $T\ra 0$ limit for the partition function and the quantum Hurwitz numbers.
The following gives the two leading terms in weighted sums of functions on $\fM_d^{(n)}$ with weights $W_{E'(q)}$ on $\mP_d$.

\begin{theorem}
	\label{thm:FullExpansion}
	For any function $g\colon\fM_d^{(n)}\longrightarrow\rr$ we have
	\begin{align}
		\sum_{k=1}^d \sum_{(\mu^{(1)},\ldots,\mu^{(k)})\in\fM_{d,k}^{(n)}}& g\rb{ \mu^{(1)}, \ldots, \mu^{(k)} } W_{E'(q)}\rb{ \mu^{(1)},\ldots,\mu^{(k)} }  \\
		 &= q^d \sum_{\ell^*(\mu^{(1)})=d} g\rb{\mu^{(1)}}  + q^{d+1} \sum_{\substack{\ell^*(\mu^{(1)})=d-1\\ \ell^*(\mu^{(2)})=1}} g\rb{\mu^{(1)},\mu^{(2)}}  + O\rb{q^{d+2}}.
	\end{align}
\end{theorem}

\par\noindent In particular,  we obtain the following leading terms in the $T=0$ expansion of
 the partition function. 

\begin{corollary}
	As $q\rightarrow 0^+$,
	\label{cor:PartitionFunction}
	\begin{align}
		\pfp & = p(d) q^d + p(d-1) q^{d+1} + O\rb{q^{d+2}}.
	\end{align}
\end{corollary}
(Recall that $p(d)$ denotes the number of integer partitions of $d$.)

For the zero temperature expansions of simple quantum Hurwitz numbers, we have the following leading terms

\begin{corollary}
	\label{cor:QuantumWeighted}
	For any $\mu,\nu\in\mP_n$ we have, as $q\rightarrow 0^+$,
	\begin{align}
		H^d_{E'(q)}(\mu,\nu)  = q^{d}\sum_{\substack{ \mu^{(1)} \in\fM_{d,k}^{(n)} \\  \ell^*(\mu^{(1)})=d}}
 H\rb{\mu^{(1)},\mu,\nu}+ q^{d+1}\sum_{\substack{ \mu^{(1)}, \,\mu^{(2)} \in\fM_{d,k}^{(n)} \\  \ell^*(\mu^{(1)})+  \ell^*(\mu^{(2)})=d}}  H\rb{\mu^{(1)},\mu^{(2)},\mu,\nu}+O\rb{q^{d+2}}.
	\end{align}
	\end{corollary}


\section{Proofs}
\label{sec:proofs}

\subsection{Classical and semiclassical asymptotics}

\begin{proof}[Proof of Lemma \ref{lem:countColengths}]
		Consider the function $f\colon\mP_{n,n-\ell}\longrightarrow \mP_\ell$ defined as follows. Let $\lambda\in\mP_{n,n-\ell}$, then the first column of the Young diagram of $\lambda$ has $n-\ell$ boxes. Remove these to obtain a partition $\nu:=f(\lambda)$ of $\ell$. This function has an inverse: for $\nu\in\mP_\ell$ simply add a new column with $n-\ell$ to the left of the Young diagram of $\nu$. Since $n-\ell\geq \ell$ by assumption the result is the Young diagram of an integer partition $\lambda:=f^{-1}(\nu)$: it is easy to see that $\lambda\in\mP_n$ and that $\ell(\lambda)=n-\ell$.
	\end{proof}

We only detail the proofs of the results from Sections \ref{sec:scl} and \ref{sec:ztl} for the case $E'(q)$. The corresponding results for $E(q)$ and $H(q)$ follow analogously, using \eqref{W_E_q} and \eqref{W_H_q}. The proofs all rely on the following asymptotic expansion of $\Phi_{E'(e^{-\ep})}$ as $\ep\longrightarrow 0$:

\begin{lemma}
	\label{lem:fundamental}
	let $x_1,\ldots, x_m\in\Zbb_{>0}$. Then, as $\ep\longrightarrow 0$
	\begin{align}
		\Phi_{E'(e^{-\ep})}\rb{x_1,\ldots,x_m} &= \ep^{-m} \sum_{\sigma\in S_m} \ab{\frac{ 1  }{ \prod_{j=1}^m \sum_{i=1}^j x_{\sigma(i)} }   - \frac{ \ep }2 \, \sum_{r=1}^m \frac{\sum_{i=1}^r x_{\sigma(i)}}{\prod_{j=1}^m \sum_{i=1}^j x_{\sigma(i)} }}  +O\rb{\ep^{2-m}} 
	\end{align}
\end{lemma}

\begin{proof}
	A direct computation yields
	\begin{align}
		\Phi_{E'(e^{-\ep})}\rb{x_1,\ldots,x_m} & = \sum_{\sigma\in S_m} \prod_{j=1}^m \rb{e^{\epsilon\sum_{i=1}^j x_{\sigma(i)}}-1}^{-1}\\
		& = \sum_{\sigma\in S_m} \prod_{j=1}^m \frac{\ep^{-1}}{\sum_{i=1}^j x_{\sigma(i)}} \rb{1 + \frac{\ep}2 \sum_{i=1}^j x_{\sigma(i)} + O(\ep^2)}^{-1}\\
		& = \ep^{-m}\, \sum_{\sigma\in S_m} \prod_{j=1}^m \rb{ \frac{1}{\sum_{i=1}^j x_{\sigma(i)}}  - \frac{\ep}2  + O(\ep^2)}\\
		& = \ep^{-m}\, \sum_{\sigma\in S_m} \rb{ \prod_{j=1}^m \frac{1}{\sum_{i=1}^j x_{\sigma(i)}} - \frac{\ep}2 \sum_{j=1}^m \frac{\sum_{i=1}^j x_{\sigma(i)}}{\prod_{r=1}^m\sum_{i=1}^r x_{\sigma(i)}}  + O(\ep^2)}
		\end{align}
	as claimed.
\end{proof}

\par\noindent By considering the highest order terms, it follows immediately that,  letting 
	\begin{align}
		d=\sum_{r=1}^m x_r&\geq m,
		\intertext{we have}
		\lim_{\ep\to 0} \ep^d \Phi_{E'(e^{-\ep})} \rb{x_1,\ldots,x_m} & = \begin{cases}
			\frac1{d!} \quad\quad & \text{if } d=m\\
			0 & \text{if } d>m.
		\end{cases}
	\end{align}

\par\noindent This completes the proof of Theorem \ref{thm:Downstairs} and hence also Corollary \ref{cor:Upstairs}. Setting $q=e^{-\ep}$ and considering additionally the terms of order $\ep^{1-d}$ gives Theorem \ref{prop:weightSC}.

\par\noindent Moreover we obtain the following intermediate result:

\begin{proposition}
	\label{lem:f}
	For any function $f\colon \fM_{d}^{(n)}\longrightarrow\Rbb$,
	\begin{align}
		\label{eq:f}
		\ep^{d} \sum_{\fM_d^{(n)}} &f\rb{\mu^{(1)},\ldots,\mu^{(k)}} W_{E'(q)}\rb{\mu^{(1)},\ldots,\mu^{(k)}} =\frac{1}{d!}  f(\underbrace{\bm2,\ldots,\bm2}_{\text{d }times}) \\
		& \quad + \frac{\ep}{(d-1)!}\,\ab{f\rb{ \underbrace{\bm 2,\ldots,\bm 2}_{d-1\text{ times}}, \bm 3 } + f\rb{ \underbrace{\bm 2,\ldots,\bm 2}_{d-1\text{ times}} , \bm{2^2}} - \,\frac{d+1}{4} f(\underbrace{\bm2,\ldots,\bm2}_{\text{d }times})} \\
		& \quad + O\rb{\ep^{2}}
	\end{align}
	where we recall that $\bm 2 = (1^{n-1},2)$ and $\bm 3 =( 1^{n-3},3)$ and $\bm {2^2} = (1^{n-4},2^2)$. 
\end{proposition}

\begin{proof}
	From Lemma \ref{lem:fundamental} it follows that $w_{E'(q)}(\lambda)$ contributes terms of order $\ep^{-\ell(\lambda)}$ and lower. Thus the only terms in \eqref{eq:f} that are not $o(\ep^{-d+1})$ correspond to elements $(\mu^{(1)},\ldots,\mu^{(k)})$ of $\fM_d^{(n)}$ such that $\lambda=\lambdamap(\mu^{(1)},\ldots,\mu^{(k)})$ has length $d$ or $d-1$, i.e. $\lambda\in\set{\bm{1},\bm{2}}$. (recalling once more the notation from Definition \ref{def:notationP}). Therefore,
	\begin{align}
		\label{eq:first}
		\sum_{\fM_d^{(n)}} &f\rb{\mu^{(1)},\ldots,\mu^{(k)}} W_{E'(q)}\rb{\mu^{(1)},\ldots,\mu^{(k)}} = \frac{p(\bm 1)}{\abs{\aut(\bm 1)}} \Phi_{e^{-\ep}}(1,\ldots,1) \sum_{\Lambda_d^{-1}(1^d)} f\rb{\mu^{(1)},\ldots,\mu^{(k)}} \\
		\label{eq:second}
		& \quad + \frac{p(\bm 2)}{\abs{\aut(\bm 2)}} \Phi_{e^{-\ep}}(2,1,\ldots,1) \sum_{\Lambda_d^{-1}(\bm2)} f\rb{\mu^{(1)},\ldots,\mu^{(k)}}  + O(\ep^{2-d}).
	\end{align}
	We first deal with the term in \eqref{eq:first}: $p(\bm{1})=1$ and $|\aut(\bm{1})|=d!$. Further, by Lemma \ref{lem:fundamental},
	\begin{align}
		\label{eq:413}
		\Phi_{E'(q)} (\underbrace{1,\ldots,1}_{d\text{ times}}) & = \ep^{-d} \sum_{\sigma\in S_d} \rb{\frac1{d!} - \frac\ep2 \frac{d(d+1)/2}{d!} +O\rb{\ep^2} } = \ep^{-d} \rb{1- \ep \,\frac{d(d+1)}{4 }} + O\rb{\ep^{-d+2}}.
	\end{align}
	For the terms in \eqref{eq:second}: $p(\bm{2})=p(2)=2$ and $\aut(\bm{2})=(d-1)!$. This time we only need the first order approximation of Lemma \ref{lem:fundamental}, and obtain
	\begin{align}
		\Phi_{E'(e^{-\ep})} (2,\underbrace{1,\ldots,1}_{d-2\text{ times}}) & = \left.\ep^{-d+1} \sum_{\sigma\in S_{d-1}} \rb{\prod_{j=1}^{d-1} \sum_{i=1}^{j} x_{\sigma(i)}}^{-1} \right\vert_{x=(2,1,\ldots,1)}  +O(\ep^{-d+2})
\end{align}
If $x=(2,1,\ldots,1)$ then we have, for $j\in\set{1,\ldots,d-1}$ and $\sigma\in S_{d-1}$,
\begin{align}
	\sum_{i=1}^j x_{\sigma(i)} &= \begin{cases}
		j+1  \quad & \text{if } j<\sigma^{-1}(1)\\
		j & \text{otherwise,}
	\end{cases}
	\intertext{and therefore}
	\left.\sum_{\sigma\in S_{d-1}} \rb{\prod_{j=1}^{d-1}\sum_{i=1}^{j} x_{\sigma(i)}}^{-1} \right\vert_{x=(2,1,\ldots,1)}  & =  \sum_{\sigma\in S_{d-1}} \rb{\prod_{j=1}^{\sigma^{-1}(1)-1} j}^{-1}\rb{\prod_{j=\sigma^{-1}(1)}^{d-1} (j+1) }^{-1} \\
	&= \sum_{\sigma\in S_{d-1}}\frac{\sigma^{-1}(1)}{d!} =\frac{1}{d!} \sum_{r=1}^{d-1} \sum_{\sigma^{-1}(1)=r}\ r \\
	&= \frac{(d-2)!}{d!}\cdot \frac{d(d-1)}2=\frac12 .
	\end{align}	
	It follows that
	\begin{align}
		\label{eq:418}
		\Phi_{E'(e^{-\ep})} (2,\underbrace{1,\ldots,1}_{d-2\text{ times}}) & = \frac12 \ep^{-d+1} + O(\ep^{-d+2}).
	\end{align}
	Substituting \eqref{eq:413} and \eqref{eq:418} into \eqref{eq:first} and \eqref{eq:second} gives
	\begin{align}
		\sum_{\fM_d^{(n)}} f\rb{\mu^{(1)},\ldots,\mu^{(k)}}& W_{E'(q)}\rb{\mu^{(1)},\ldots,\mu^{(k)}} 
		=	\frac{ \ep^{-d}}{d!}\rb{1-\ep \frac{d(d+1)}4} f\rb{\underbrace{\bm2,\ldots,\bm2}_{d\text{ times}}}\\
		& \quad + \frac{ \ep^{-d+1} }{(d-1)!}\rb{f(\bm 3,\underbrace{\bm 2,\ldots,\bm 2}_{d-2\text{ times}}) + f(\bm{2}^2,\underbrace{\bm 2,\ldots,\bm 2}_{d-2\text{ times}})}	+ O(\ep^{-d+2})	
	\end{align}
as required.%

\end{proof}

Choosing $f\rb{\mu^{(1)},\ldots,\mu^{(k)}} = H\rb{\mu^{(1)},\ldots,\mu^{(k)},\mu,\nu}$ gives Theorem \ref{thm:WeightedHurwitz}. On the other hand by setting $f\rb{\mu^{(1)},\ldots,\mu^{(k)}}=1$ we obtain
\begin{align}
	\pf d{e^{-\ep}} &= \frac{\ep^{-d}}{d!} + \ep^{1-d}\rb{\frac1{(d-1)!} -\frac{d(d+1)}{4d!}}\\
	&= \frac{\ep^{-d}}{d!} + \ep^{1-d}\, \frac{3-d}{4(d-1)!} + O\rb{\ep^{2-d}}
\end{align}
and so we have proved Proposition \ref{cor:pfSC}.

\ \\


\subsection{Zero-temperature limit: leading order terms}

\begin{proof}[Proof of Theorem \ref{thm:DownstairsZeroTemp}]

	For any $\lambda\in\mP_d$,
	\begin{align}
		w_{E'(q)}(\lambda) & = \frac1{\aut(\lambda)} \sum_{\sigma \in S_{\ell(\lambda)}} \prod_{j=1}^{\ell(\lambda)} \frac{q^{\sum_{i=1}^j \lambda_{\sigma(i)} }}{ 1-q^{\sum_{i=1}^j \lambda_{\sigma(i)} } }\\
		& = \frac1{\aut(\lambda)} \sum_{\sigma \in S_{\ell(\lambda)}} \prod_{j=1}^{\ell(\lambda)} q^{\sum_{i=1}^j \lambda_{\sigma(i)} } \rb{1+O\rb{q}}\\
		& = \frac1{\aut(\lambda)} \sum_{\sigma \in S_{\ell(\lambda)}}  q^{\sum_{i=1}^{\ell(\lambda)} (\ell(\lambda)-i+1)\lambda_{\sigma(i)} } \rb{1+O\rb{q}}\\
		& = \frac1{\aut(\lambda)} \sum_{\sigma \in S_{\ell(\lambda)}}  q^{\sum_{j=1}^{\ell(\lambda)} j\lambda_{\sigma(j)} } \rb{1+O\rb{q}}
	\end{align}
	Since $\lambda_1\geq \lambda_2\geq\ldots\geq \lambda_{\ell(\lambda)}$ and $q$ is small, the  sum above is dominated by the contribution when $\sigma$ is the identity permutation. In particular we obtain
	\begin{align}
		\label{eq:FirstOrderW}
		w_{E'(q)}(\lambda) & = \frac1{\aut(\lambda)}\,  q^{\sum_{j=1}^{\ell(\lambda)} j\lambda_{j} } \rb{1+O\rb{q}}.
	\end{align}
	Thus, in the limit $q\to 0^+$ the dominant weight will be given by $\lambda$ such that $\sum_j j\lambda_j$ is minimised, i.e. $\lambda=(d)$. This completes the proof.

\end{proof}


\subsection{Zero-temperature limit: higher-order corrections}

\par\noindent Having established the leading order result we urn to the next  order corrections. 
It follows from \eqref{eq:FirstOrderW} that
\begin{align}
	\sum_{\substack{\lambda\ne (d)\\\lambda\ne (d-1,1)}} w_{E'(q)}(\lambda ) & = O\rb{q^{d+2}}.
\end{align}
Further,
\begin{align}
	w_{E'(q)}((d)) & = \frac1{\aut((d))}\ \frac{q^d}{1-q^d} = q^d+O\rb{q^{2d}},\\
	w_{E'(q)}((d-1,1)) &= \frac{q^{d-1}}{1-q^d}\,\frac{q^d}{1-q^d} + \frac q{1-q}\, \frac{q^d}{1-q^d}\\
	& = \frac{q^{d+1}}{(1-q)(1-q^d)} + O\rb{q^{2d-1}}\\
	& = q^{d+1}\rb{1+q+O\rb{q^2}}\rb{1+q^d+O\rb{q^{2d}}} \\
	&= q^{d+1} +  q^{d+2} +O\rb{q^{d+3}}.
\end{align}
For any $f\colon\mP_d\longrightarrow\rr$ we therefore have
\begin{align}
	\label{eq:FullExpansionDownstairs}
		\sum_{\lambda\in\mP_d} f(\lambda) w_{E'(q)}(\lambda) & = f((d))\ p(d)  q^d + f((d-1),(1))p(d-1)  q^{d+1} + O\rb{q^{d+2}}.
\end{align}
Theorem \ref{thm:FullExpansion} now follows from \eqref{eq:FullExpansionDownstairs} and the discussion in Section 2. 

By choosing $g$ to be identically equal to 1 we obtain Corollary \ref{cor:PartitionFunction}, whereas choosing, for fixed $(\mu,\nu)$,
\begin{align}
	g(\mu^{(1)},\ldots,\mu^{k})=H(\mu^{(1)},\ldots,\mu^{k},\mu,\nu),
\end{align}   Corollary \ref{cor:QuantumWeighted} is just Theorem \ref{thm:FullExpansion} for this particular case.

\bigskip
\bigskip
\noindent 
\small{ {\it Acknowledgements.}  The authors would  like to thank G. Borot and A.Yu.~Orlov for helpful discussions.
The work of JH was partially supported by the Natural Sciences and Engineering Research Council of Canada (NSERC) and the Fonds de recherche du Qu\'ebec, Nature et technologies (FRQNT). JO was partially supported by a CRM-ISM postdoctoral fellowship and a Horizon postdoctoral fellowship.}
\bigskip

\newcommand{\arxiv}[1]{\href{http://arxiv.org/abs/#1}{arXiv:{#1}}}

\bigskip
\noindent

\end{document}